\documentclass[twocolumn]{IEEEtran}
\usepackage{graphicx}
\usepackage{cite}
\usepackage{picinpar}
\usepackage{amsmath}
\usepackage{url}
\usepackage{flushend}
\usepackage[latin1]{inputenc}
\usepackage{colortbl}
\usepackage{soul}
\usepackage{multirow}
\usepackage{pifont}
\usepackage{color}
\usepackage{alltt}
\usepackage[hidelinks]{hyperref}
\usepackage{enumerate}
\usepackage{siunitx}
\usepackage{breakurl}
\usepackage{epstopdf}
\usepackage{pbox}
\usepackage[vlined,linesnumberedhidden,ruled,resetcount]{algorithm2e}

\usepackage{graphics} 
\usepackage{graphicx} 
\usepackage{epsfig} 
\usepackage{enumitem} 

\usepackage{enumerate}
\usepackage{amsmath} 
\usepackage{amssymb}  

\usepackage{amsthm}

\newcommand{\nextnr}{\stepcounter{AlgoLine}\ShowLn}

\newtheorem{prop}{Proposition}
\newtheorem{rem}{Remark}

\begin{document}
\title{Cylindrical Battery Fault Detection under Extreme Fast Charging: A Physics-based Learning Approach}

\author{
	\vskip 1em
	
    {
    Roya Firoozi, Sara Sattarzadeh, and Satadru Dey
	}

\thanks{
       {This work was supported by National Science Foundation under Grant No. 1908560 and 2050315. The opinions, findings, and conclusions or recommendations expressed are those of the author(s) and do not necessarily reflect the views of the National Science Foundation.
        
        R. Firoozi is with the Department of Mechanical Engineering, University of Berkeley, CA, USA. (e-mail: royafiroozi@berkeley.edu).
        
        S. Dey and S. Sattarzadeh are with the Department of Mechanical Engineering, The Pennsylvania State University, PA 16802, USA. (e-mail: skd5685@psu.edu, sfs6216@psu.edu). 
 }
 }
 }

\maketitle
	
\begin{abstract}
High power operation in extreme fast charging significantly increases the risk of internal faults in Electric Vehicle batteries which can lead to accelerated battery failure. Early detection of these faults is crucial for battery safety and widespread deployment of fast charging. In this setting, we propose a real-time {detection} framework for battery voltage and thermal faults. A major challenge in battery fault detection arises from the effect of uncertainties originating from sensor inaccuracies, nominal aging, or unmodelled dynamics. Inspired by physics-based learning, we explore a detection paradigm that combines physics-based models, model-based detection observers, and data-driven learning techniques to address this challenge. Specifically, we construct the {detection} observers based on an experimentally identified electrochemical-thermal model, and subsequently design the observer tuning parameters following Lyapunov's stability theory. Furthermore, we utilize Gaussian Process Regression technique to learn the model and measurement uncertainties which in turn aid the {detection} observers in distinguishing faults and uncertainties. Such uncertainty learning essentially helps suppressing their effects, potentially enabling early detection of faults. We perform simulation and experimental case studies on the proposed fault {detection} scheme verifying the potential of physics-based learning in early detection of battery faults.
\end{abstract}



\definecolor{limegreen}{rgb}{0.2, 0.8, 0.2}
\definecolor{forestgreen}{rgb}{0.13, 0.55, 0.13}
\definecolor{greenhtml}{rgb}{0.0, 0.5, 0.0}

\section{Introduction}


{Extreme} {fast charging typically refers to charging a battery to $80\%$ of its capacity within $10$ minutes \cite{ahmed2017enabling}.} The high power requirement for fast charging
significantly increases the risk of internal faults that affect voltage and temperature dynamics. 
In this research, we address this particular issue of battery safety under fast charging. 
Existing works on battery fault diagnostics can be broadly classified into two categories: \textit{model-based techniques} and \textit{data-driven techniques}. Examples of such model-based techniques include \cite{dey2017model,dey2015diagnostic,marcicki2010nonlinear,wei2019lyapunov,xiong2019online}. {On the other hand, learning-based techniques can be useful in circumventing the need for accurate models}. For example, data-driven battery diagnostic approaches are presented in \cite{saha2008uncertainty,zhao2017fault,zhang2018long,li2020battery,lee2021convolutional}.


However, there are certain limitations associated with these aforementioned approaches: (i) model-based approaches rely on accurate models identification of which can be a cumbersome process; (ii) most of the proposed model-based approaches utilize phenomenological type models which have limited capabilities to capture physical failure modes; (iii) data-driven approaches require significant amount of data encompassing probable fault scenarios, which may not be available in real-world settings; (iv) data-driven approaches might exhibit limitations on diagnosing unforeseen faults which were not present in training data; (v) finally, most of these approaches do not present a systematic way to distinguish fault from uncertainties, which is a fundamental challenge in any fault diagnosis algorithm design. {Specific to battery fault detection, the main challenge arises from effect of other (non-faulty) phenomena that cause voltage and temperature deviation. In case of voltage, these phenomena include voltage sensor noise, drift or bias, effect of unmodelled behavior such as electrolyte dynamics (while using a model-based detection approach), and effect of nominal aging which is generally a slower process than faults. Inaccurate heat generation model, temperature sensor noise, drift or bias, and thermal influence of adjacent cells contribute to similar effects in case of temperature. These phenomena, which can be combined as \textit{uncertainties}, prohibit early detection of smaller voltage and thermal faults. Hence, it is imperative to distinguish the effect of \textit{uncertainties} from faults to enable early detection of smaller faults.}

In light of the aforementioned limitations, our main contribution is the following: \textit{Departing from the existing model-only and data-only diagnostic approaches, we explore a different diagnostics paradigm that combines physics-based models, model-based detection observers, and data-driven learning techniques.} Our approach is inspired by model-based learning framework \cite{berkenkamp2015safe,recht2019tour}. Our approach can be summarized as follows: (i) we start with a coarsely identified uncertain physics-based model circumventing the cumbersome process of accurate model identification; (ii) we apply an online data-driven learning technique, namely Gaussian Process Regression \cite{rasmussen2003gaussian}, to learn the uncertainty functions in real-time thereby compensating for the model and measurement uncertainties; and (iii) based on the coarsely identified model and aided by the learned uncertainty functions, we design detection observers following Lyapunov's stability theory to diagnose battery faults. {The main advantages of the proposed approach are the following: (i) it can help overcome the limitations of model-based approaches (i.e. need for accurate model and presence of model uncertainties); and (ii) as opposed to data-driven techniques, it does not need for vast data encompassing all possible fault scenarios.}

The rest of the paper is organized as follows. Section II describes the phsyics-based battery models. Section III details the proposed fault detection framework. Section IV presents experimental and simulation results with discussion. Finally, section V concludes the work.

\section{Battery Electrochemical-Thermal Model}
In this section, we discuss the battery model adopted for this work. Specifically, we focus on electrochemical and thermal dynamics of batteries. We adopt the Single Particle Model (SPM) framework to capture electrochemical dynamics and focus on anode dynamics {\cite{santhanagopalan2006online,7004795}}:
\begin{align}
    & \frac{\partial c_a(x,t)}{\partial t} = \frac{D}{x^2} \frac{\partial}{\partial x} \Big( {x^2}\frac{\partial c_a(x,t)}{\partial x} \Big), \label{spm1}\\
    & \frac{\partial c_a(x,t)}{\partial x}\Big|_{x=0} = 0, \quad \frac{\partial c_a(x,t)}{\partial x}\Big|_{x=X} = \frac{-I(t)}{a_aFDA_aL_a}, \label{spm2}
\end{align}
where $x$ is the radial coordinate of the particle in {m}, $t$ is the time in {s}, $c_a$ is the Lithium concentration along the particle radius in {mol/m$^3$}, $D$ is the anode diffusion coefficient in {m$^2$/s}, $X$ is radius of the particle in {m}, $I$ is the applied current in {A} with $I>0$ indicating discharging, $a_a$ is the anode specific surface area in {m$^2$/m$^3$} which is computed as $a_a = 3\epsilon_a/X$ where $\epsilon_a$ is the active material volume fraction, $F$ is the Faraday's constant in {C/mol}, $A_a$ is the anode current collector area in {m$^2$}, and $L_a$ is the anode thickness in {m}. {In SPM framework, electrodes are approximated as spherical particles. Based on such approximation, \eqref{spm1} describes the solid-phase diffusion of Lithium ions in the anode, governed by the volume-averaged current acting on the on the boundary as given in \eqref{spm2}.}

The battery terminal voltage can be expressed as {\cite{7004795}}:
\begin{align}
    & V_{term}(t) = U_c\Big(\alpha_1 c_a(X) + \alpha_2\Big) + U_a\Big(c_a(X)\Big) - R_{b}I(t) \nonumber\\
    & + \frac{RT}{\alpha_c F}\sinh^{-1}\Big({\frac{I(t)}{2a_c A_c L_c i_{0c}}}\Big) + \frac{RT}{\alpha_a F}\sinh^{-1}\Big({\frac{I(t)}{2a_a A_a L_a i_{0a}}}\Big) , \label{termV}
\end{align}
where $V_{term}$ is the terminal voltage in {V}; $U_c(.)$ and $U_a(.)$ are the open circuit potential maps of cathode and anode, respectively; $\alpha_1 = -(\epsilon_aA_aL_a)/(\epsilon_cA_cL_c)$ and $\alpha_2 = m_{Li}/(\epsilon_cA_cL_c)$ with $m_{Li}$ being total moles of Lithium in the cell; $R$ is the universal gas constant in {J/mol-K}, $T$ is the average cell temperature in {K}, $\alpha_c$ and $\alpha_a$ are unitless charge transfer coefficients of cathode and anode, respectively; $a_c$ in {m$^2$/m$^3$}, $A_c$ in {m$^2$}, and $L_c$ in {m} are the specific area, current collector area, and thickness of the cathode, respectively; $i_{0c}$ and $i_{0a}$ are the exchange current densities of cathode and anode in {A/m$^2$}, respectively; $R_b$ is the internal resistance of the cell in {$\normalfont{\Omega}$}. {In \eqref{termV}, the first two terms represent the thermodynamic potential of the electrodes, the third term captures the effect of internal Ohmic resistances, and the last two terms represent the electric overpotential of the electrodes.}

We adopt the following thermal model {\cite{ALHALLAJ19991}}:
\begin{align}
    & \rho C_p \frac{\partial T(y,t)}{\partial t} = k \frac{\partial^2 T(y,t)}{\partial y^2}+ k \frac{1}{y}\frac{\partial T(y,t)}{\partial y} + \frac{\dot{Q}(t)}{V_b}, \label{therm1}\\
    & \frac{\partial T(y,t)}{\partial y}\Big|_{y=0} = 0, \quad \frac{\partial T(y,t)}{\partial y}\Big|_{y=Y} = -\frac{h}{k}(T(Y)-T_{\infty} ), \label{therm2}
\end{align}
where $t$ is the time in $s$, $y$ is the radial coordinate of the cell in $m$, $Y$ is the cell radius in $m$, $\rho$ is the cell density in {kg/m$^3$}, $C_p$ is the specific heat in {J/kg-K}, $k$ is the thermal conductivity in {W/m-K}, $h$ is the convection coefficient in {W/m$^2$-K}, $T_{\infty}$ is the cooling/ambient temperature, $V_b$ is the cell volume in {m$^3$}, and $\dot{Q}$ is the heat generation term computed by {\cite{ALHALLAJ19991}}
\begin{align}
    & \dot{Q}(t) = I(t)\Big\{ U_c\Big(\alpha_1 c_a(X) + \alpha_2\Big) + U_a\Big(c_a(X)\Big) - V_{term}(t)\Big\}. \label{heat}
\end{align}
{The model \eqref{therm1} is derived using energy balance principle and captures the spatio-temporal temperature dynamics along the cell radius in a cylindrical coordinate setting \cite{ALHALLAJ19991}. In \eqref{therm2}, the first boundary condition represents the zero temperature gradient in the center whereas the second boundary condition captures the convective heat transfer with the environment. The heat generation model \eqref{heat} captures the irreversible heat due to electrode polarization \cite{guo2011thermal}.} {In this work, we have ignored the reversible heat generation effect in \eqref{heat} to keep the heat generation model simple enough for identification and real-time computation purposes. Typically, the reversible heat model requires additional parameter information (e.g. entropic coefficients) which again leads to cumbersome parameter identification. The effect of such reversible heat is later captured by the uncertainties in thermal model.}

Both the nominal electrochemical and thermal models are in Partial Differential Equation (PDE) form. We convert these PDEs to a set of Ordinary Differential Equations (ODEs), and subsequently formulate a state-space model. First, considering the electrochemical PDE model \eqref{spm1}-\eqref{spm2}, we discretize the particle radius into $N+1$ nodes with each node's Lithium concentration defined as $c_i = c_a(i\delta_x), i = \{0,1,2,...,N\}$ where $\delta_x = X/N$ is the spatial difference between two adjacent nodes. Next, we follow the method of lines approach for PDE to ODE conversion 
leading to the following set of ODEs for each node:
\begin{align}
    & \dot{c}_1 = -2\gamma_1 c_1 + 2\gamma_1c_2,\nonumber\\
    & \dot{c}_i = (1-\frac{1}{i})\gamma_1c_{i-1} -2\gamma_1c_i + (1+\frac{1}{i})\gamma_1c_{i+1},\nonumber\\
    & \dot{c}_N = (1-\frac{1}{N})\gamma_1)c_{N-1} -(1-\frac{N}{i})\gamma_1c_{N} + \beta_1 I,\label{ss-echem}
\end{align}
with $c_0 = c_1$ from the boundary condition, $i \in \{2,...,N-1\}$, $\gamma_1 = D/\delta_x^2$, and $\beta_1 = -(1/a_aFA_aL_a\delta_x)(1+1/N)$. In a similar manner, we can discretize the thermal PDE equation \eqref{therm1}-\eqref{therm2} with the nodes $T_j = T(j\delta_y), j = \{0,1,2,...,M\}$ where $\delta_y = Y/M$ and arrive at the following set of ODEs:
\begin{align}
    & \dot{T}_1 = -1.5\gamma_2 T_1 + 1.5 \gamma_2 T_2 + \dot{Q}/V_b,\nonumber\\
    & \dot{T}_j = (1-\frac{1}{2i})\gamma_2c_{j-1} -2\gamma_2T_i + (1+\frac{1}{2i})\gamma_2T_{i+1} + \dot{Q}/V_b,\nonumber\\
    & \dot{T}_M = (1-\frac{1}{2M})\gamma_2T_{M-1} -(1-\frac{M}{i})\gamma_1T_{M} + \dot{Q}/V_b + \beta_2 T_{\infty},\label{ss-therm}
\end{align}
with $T_0 = T_1$ from the boundary condition, $j \in \{2,...,M-1\}$, $\gamma_2 = k/\delta_y^2$, and $\beta_2 = -2\gamma_2 + \gamma_2(1+1/2M)(1-\delta_yh/k)$. 


\subsection{Failure Modes}

In the extreme fast charging applications, the battery cell goes through higher amount of stress than the normal operating scenarios. Accordingly, the probability of the failure is higher than the nominal operation. {Some of the critical failure modes are: active material and Lithium inventory loss, undesirable side reactions, electrode fracture, electrolyte decomposition, Lithium plating, internal and external short circuits, electrode fracture, separator puncture, and abnormal heating. Further details of failure modes can be found in \cite{hu2020advanced}.} Specific to fast charging, some potential battery faults include Lithium plating, thermal runaway, and mechanical degradation such as stress-induced cathode cracking and separation of current collector and electrode \cite{tomaszewska2019lithium,ahmed2017enabling}. In this work, we focus on two types of faults to illustrate our proposed framework. First type is \textit{voltage faults} that mainly affect terminal voltage. Based on their origins, these faults can be classified into following categories: (i) faults originating from electrochemical side reactions such as Lithium plating, Solid Electrolyte Interphase (SEI) growth, and electrolyte decomposition; and (ii) faults originating from eletrical anomalies such as current leaks, external and internal short short circuits. Another way of classifying voltage faults is based on their dynamic response: (i) incipient type voltage faults that gradually show up in a longer period, e.g. SEI growth, and (ii) abrupt type voltage faults that posses faster dynamics, e.g. Lithium plating. Second type is \textit{abnormal heating faults} that mainly affect thermal dynamics. Some potential sources of abnormal heat generation are unwanted side reactions, seperator failure, overcharging, and external shock or puncture \cite{doughty2012general,bandhauer2011critical}. 

\subsection{State-Space Model with Faults and Uncertainties}
We define the state vectors $z_1 = [c_1,c_2,...,c_N]^T$ and $z_2 = [T_1,T_2,...,T_M]^T$, the inputs $u_1 = I$ and $u_2 = T_{\infty}$, and the outputs $y_1=V_{term}$ and $y_2=T_M$. Applying Euler's discretization and linearization on \eqref{termV}, we formulate the following discrete-time state space model from \eqref{ss-echem}, \eqref{termV}, \eqref{heat}, and \eqref{ss-therm}:
\begin{align}
    & {{z}_1}_{t+1} = A_1 {z_1}_t + B_1 {u_1}_t,\label{ss-echem-21}\\
    & {y_1}_t = C_1 {z_1}_t + D_1{u_1}_t + {\omega_V}_t + {\Delta_{V}}_t,\label{ss-echem-2}\\
    & {{z}_2}_{t+1} = A_2 {z_2}_t + f_2(z_1,y_1,u_1) + B_2 {u_2}_t + {\Delta_{T}}_t,\label{ss-therm-21}\\
    & {y_2}_t = C_2 {z_2}_t + {\omega_T}_t,\label{ss-therm-2}
\end{align}
where the subscript $t$ denotes the time index, $A_1 \in \mathbb{R}^{N\times N}$, $B_1\in \mathbb{R}^{N\times 1}$, $A_2 \in \mathbb{R}^{M\times M}$, $B_2\in \mathbb{R}^{M\times 1}$, $C_2 = [0,0,..,0,1] \in \mathbb{R}^{1\times M}$, $C_1 \in \mathbb{R}^{1\times N}$, $D_1 \in \mathbb{R}$, and the nonlinear function $f_2(.)$ is derived from \eqref{heat}. Furthermore, the term $\omega_V$ captures the model and measurement uncertainties as well as linearization error in electrochemical dynamics, $\omega_T$ captures the thermal model and measurement uncertainties, $\Delta_{V}$ captures the {voltage fault}, and $\Delta_{T}$ captures the thermal fault.

\begin{rem}\label{rem1}\normalfont
{As we will use the output measurements to design fault indicator signals in our detection scheme, we mainly focus on the uncertainties affecting the system outputs, namely, terminal voltage and surface temperature. From a physical viewpoint, the term $\omega_V$ captures the following: (i) effect of internal resistive component due to variation in electrolytic conductivity, (ii) effect of internal resistance rise due to power fade type aging, (iii) effect of electrolytic concentration states on voltage, (iv) change in the parameter $m_{Li}$ due to loss of active material and Lithium inventory, and (v) inaccuracies in voltage sensor. Similarly, the term $\omega_T$ captures the following: (i) the effect of non-uniform heat generation and non-uniform thermal conductivity on the surface temperature prediction, (ii) inaccurate knowledge of convection coefficient affecting surface temperature prediction, (iii) thermal conductivity variation due to aging, and (iv) inaccuracies in temperature sensor.}
\end{rem}

\section{Fault Detection Framework}
The proposed fault detection scheme is shown in Fig. \ref{fault-schematic}. The scheme consists of following four subsystems that interact with each other: (i) \textit{Electrochemical Detection Observer} receives the measured signals and in turn produces {voltage residual} signal $r_V$. 
(ii) \textit{Thermal Detection Observer} receives the measured signals and produces thermal residual signal $r_T$. 
(iii) \textit{Learning Algorithm} receives the measured signals and produces estimates of the uncertainties $\omega_T$ and $\omega_V$. 
(iv) \textit{Decision Maker} receives the residual signals and makes a decision whether a fault has occurred. Next, we will discuss the details of these subsystems.

\begin{figure}
    \centering
    \includegraphics[trim = 0mm 0mm 0mm 0mm, clip, scale=0.8, width=0.95\linewidth]{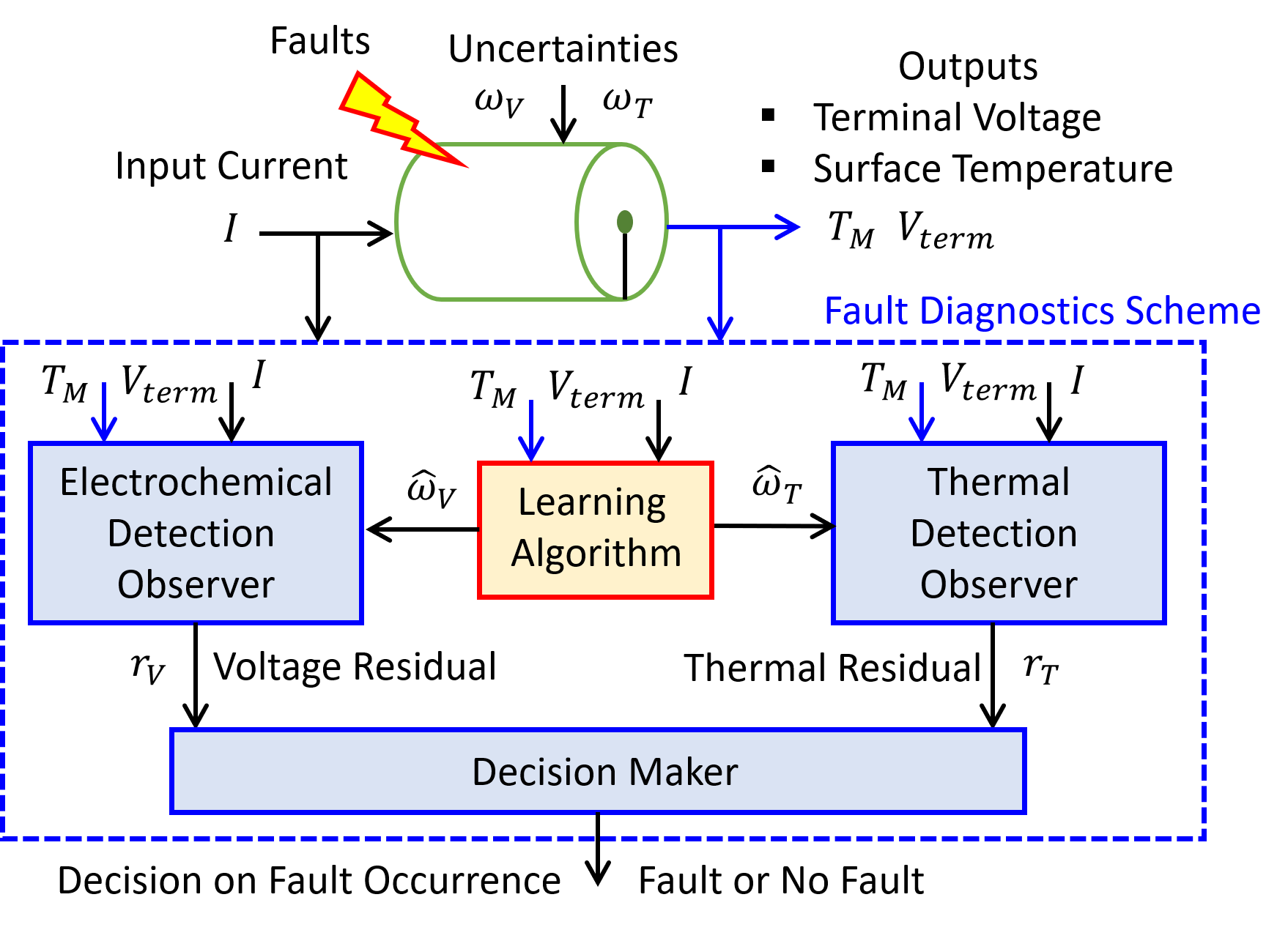}
    \caption{Fault detection scheme.}
    \label{fault-schematic}
\end{figure}

\subsection{Gaussian Process Regression based Learning}

From fault detection viewpoint, it is a crucial and challenging task to distinguish faults from uncertainties. In this scheme, we utilize Gaussian Process Regression technique to learn the uncertainties in real-time. Gaussian Process Regression technique is a non-parametric kernel-based probabilistic model, which is a well-suited method for long term learning. To update the physics-based models in real-time, we learn and update the uncertainties ${\omega_V}$ and ${\omega_T}$ of the system \eqref{ss-echem-2} and \eqref{ss-therm-2} using Gaussian Process Regression. Essentially, we follow the Algorithm \ref{algo:a} to implement the learning. Next, we describe the Gaussian Process Regression learning approach {outlined in the Algorithm \ref{algo:b}}.

\LinesNotNumbered
\begin{algorithm}
\KwIn{Measured data $I$, $V_{term}$, $T_M$}
\KwOut{Uncertainty functions $\hat{\omega}_V$, $\hat{\omega}_T$}
\nextnr
Initialize $\hat{\omega}_V(.)=0$, $\hat{\omega}_T(.)=0$.\\
\nextnr
Based on $I$, $V_{term}$ and $T_M$ data from \textit{Cycle \# 1}, learn the functions $\hat{\omega}_V(.)$ and $\hat{\omega}_T(.)$, {using Algorithm 2.}\\
\nextnr
\eIf{No fault detected in \textit{Cycle \# M}}{Use $I$, $V_{term}$ and $T_M$ data from \textit{Cycle \# $M$} to update the functions $\hat{\omega}_V(.)$ and $\hat{\omega}_T(.)$, {using Algorithm 2}\;}
   {Use $I$, $V_{term}$ and $T_M$ data from the last \textit{Cycle} before \textit{Cycle \# M} where no fault is detected\;
   Update the functions $\hat{\omega}_V(.)$ and $\hat{\omega}_T(.)$ based on that data, {{}using Algorithm 2}\;}
\nextnr
Repeat \textit{Step 3}, till battery End of Life (EOL) is reached.
\caption{{Life-time Learning Algorithm}}
\label{algo:a}
\end{algorithm}

\begin{algorithm}
{\KwIn{Training dataset $\mathcal{D}$}
\KwOut{Learned uncertainty $\hat{\omega}_{\text{new}}$ with mean $\mu$ and covariance $\sigma^2$}
\nextnr
Specify the hyperparameters $(\sigma_p^2,\mathcal{L})$ in \eqref{kernel}. \\
\nextnr
Compute the entries of matrix $\mathcal{Q}$ in (\ref{Q}) by evaluating the kernel \eqref{kernel} for the specified hyperparameters.\\
\nextnr
Use (\ref{Q}) to compute (\ref{mu_new}) and (\ref{cov_new}).\\
Return the predicted uncertainty $\hat{\omega}_{\text{new}}$ represented by the mean $\mu$ and covariance $\sigma^2$ in (\ref{mu_new}) and (\ref{cov_new}), respectively.}\\
\caption{{Gaussian Process Learning Algorithm}}
\label{algo:b}
\end{algorithm}

As explained in the Algorithm \ref{algo:a}, in order to use a cycle data as the training set for Gaussian Process Regression, the algorithm \ref{algo:a} first checks if no fault is detected in that cycle, otherwise the cycle data is discarded. So we run the learning algorithm \ref{algo:b} only for no-fault cycles where ${\Delta_{V}}_t$ is zero. {Two separate Gaussian Process models are trained to learn voltage ${\omega_V}$ and temperature ${\omega_T}$ uncertainties.}
To train a Gaussian Process Regression model, the first step is to create a training dataset. We use the previously observed {time-series} data in the most recent (past) cycle $\mathcal{D}$ as the training set. {The dataset contains a matrix of inputs $X \in \mathbb{R}^{N\times3}$ and a vector of outputs (labels) $Y \in \mathbb{R}^N$, where $N$ denotes the cycle length and 3 denotes the three types of measurement signals. The training dataset is defined as $\mathcal{D}=\{X = \{I_l, V_{{term}_l}, T_{{M}_l}\}, Y = \{\omega_l\}\}_{l=1:N}$}, where $I_l$ is the input current, $V_{{term}_l}$ is the measured voltage, and  {${\omega}_l$ is the residual between the actual measurements and the predicted values computed by the corresponding model for voltage \eqref{ss-echem-2} or temperature \eqref{ss-therm-2}.}

After creating the dataset, the next step is to train a Gaussian Process Regression model on that data. {The algorithm \ref{algo:b} outlines and summarises the steps for training a Gaussian process model}. 
Our goal is to learn a model for the nonlinear function ${\omega}_t$.
Since no prior knowledge is available for ${\omega}_t$, we choose the prior mean as zero and the covariance between any two data points $\omega(v_l)$ and $\omega(v_q)$ is defined as the squared-exponential kernel {\cite{3569}} 
\begin{equation}\label{kernel}
    k(v_l,v_q) = \sigma_{p}^2 \text{exp}(-\frac{1}{2}(v_l-v_q)^T \mathcal{L}^{-2}(v_l-v_q)),
\end{equation}
where $\sigma_{p}^2$ is the signal variance and the matrix $\mathcal{L}$ is a diagonal matrix and its diagonal elements are length scales that represent the function smoothness parameters. 
{The squared-exponential kernel is widely-used to define the covariance in time-series data and} is suitable for modeling smooth functions. Therefore, it is appropriate for our application, since the battery terminal voltage and surface temperature evolve in a sufficiently smooth manner. {The unknown hyper-parameters ($\sigma_{p}^2$, $\mathcal{L}$) associated with the kernel \eqref{kernel}, can either be specified by the user or by maximizing log-likelihood estimation on the training data or by cross-validation on the training data \cite{3569}. In this study, we consider the first option, in which the user determines constant hyperparameters by tuning and without performing optimization.}   
   
After specifying {the hyperparameters, the prior distribution is defined using \eqref{kernel}}. Then, the joint probability distribution of the new point of interest {(test point)} and the past observation data is computed as
$
\begin{bmatrix}
\omega\\
\hat{\omega}_{\text{new}}
\end{bmatrix}
\sim \mathcal{N}({0,\mathcal{Q}}),
$ 
where $\omega$ is the vector of observed data $\omega = \{\omega(v_1),...,\omega(v_N)\}$ and $\hat{\omega}_{\text{new}}$ is the function value associated with the new point of interest  $v_\text{new}$. The vector $\mathbf{0} \in \mathbb{R}^{N+1}$ is the mean and the covariance matrix is given by {\cite{3569}}
\begin{align}\label{Q}
    \mathcal{Q} = \begin{bmatrix}
\Sigma_{1:N,1:N} & \Sigma_{1:N,\text{new}}\\
\Sigma_{\text{new},1:N} & \Sigma_{\text{new},\text{new}}\\
\end{bmatrix},
\end{align}
where $\Sigma_{1:N,1:N}$ is defined by kernel \eqref{kernel} and has entries $(\Sigma_{1:N,1:N})_{lq} = k(v_l,v_q)$ for $l,q \in \{1,..,N\}$, $\Sigma_{\text{new},1:N}$ is defined as $[k(v_\text{new},v_1),...,k(v_\text{new},v_N)]$ and $\Sigma_{\text{new},\text{new}}$ is $k(v_\text{new},v_\text{new})$. The predictive posterior distribution on our point of interest is calculated as a multivariate Gaussian conditioned on the past observations is $\hat{\omega}_{\text{new}}|\mathcal{D} \sim \mathcal{N}({\mu(\omega_{\text{new}}),\sigma^2(\omega_{\text{new}})})$ {\cite{3569} with mean $\mu$ and covariance $\sigma^2$ defined as}
\begin{align}
    \mu(\omega_{\text{new}}) &= (\Sigma_{\text{new},1:N})(\Sigma_{1:N,1:N})^{-1}\omega \label{mu_new}\\
    \sigma^2(\omega_{\text{new}}) & = \Sigma_{\text{new},\text{new}} \nonumber\\
    &-(\Sigma_{\text{new},1:N})(\Sigma_{1:N,1:N})^{-1}(\Sigma_{1:N,\text{new}}) \label{cov_new}.
\end{align}

{Note that the Gaussian Process described above is presented in general form. However, to learn voltage and temperature uncertainties we use slightly different Gaussian Processes. To estimate the voltage uncertainty $\hat{\omega}_V$, 
the training dataset is $\mathcal{D}_V = \{I_l, V_{{term}_l},{\omega_V}_l\}_{l=1\colon N}$, where ${\omega_V}_l$ is computed as ${\omega_V}_l = V_{{term}_l} - (C_1 {z_1}_l + D_1{u_1}_l)$ from \eqref{ss-echem-2}. Also, to estimate the temperature uncertainty, $\hat{\omega}_T$, 
the training dataset is $\mathcal{D}_T = \{I, V_{{term}_l},T_{Ml},{\omega_T}_l\}_{l=1\colon N}$ in which $T_{Ml}$ is the measured surface temperature data and ${\omega_T}_l$ is computed as ${\omega_T}_l = T_{Ml} - C_2 {z_2}_l $ from \eqref{ss-therm-2}.}

\begin{rem}\normalfont
{As mentioned in Algorithm 1, the uncertainty models $\hat{\omega}_V(.)$ and $\hat{\omega}_T(.)$ are learned and periodically updated online as the battery cycles. Moreover, these uncertainty models capture a small subset of the physical modes as compared to the entire battery model, as clarified in Remark \ref{rem1}. Hence, the uncertainty models require lesser amount of data to be trained. Furthermore, there is no assumption on the uncertainty model structure and operating scenarios. Therefore, such learning framework can be applied to any operating conditions and the algorithm will adapt to the same.}
\end{rem}

\subsection{Design of Detection Observers}
Based on the electrochemical and thermal models \eqref{ss-echem-21}-\eqref{ss-echem-2} and \eqref{ss-therm-21}-\eqref{ss-therm-2}, we choose the following structure for the detection observers:
\begin{align}
    & \hat{z}_{1t+1} = A_1 \hat{z}_{1t} + B_1 u_{1t} + L_V (y_{1t}-\hat{y}_{1t}),\label{ss-echem-obs}\\
    & \hat{y}_{1t} = C_1 \hat{z}_{1t} + D_1 u_{1t}+ \hat{\omega}_{Vt},\label{ss-echem-obs-2}\\
    & r_{Vt} = y_{1t}-\hat{y}_{1t}, \label{ss-echem-res}\\
    & \hat{z}_{2t+1} = A_2 \hat{z}_{2t} + f_2(\hat{z}_1,y_1,u_1) + B_2 u_{2t} + L_T(y_{2t}-\hat{y}_{2t}),\label{ss-therm-obs}\\
    & \hat{y}_{2t} = C_2 \hat{z}_{2t} + \hat{\omega}_{Tt},\label{ss-therm-obs-2}\\
    & r_{Tt} = y_{2t}-\hat{y}_{2t}, \label{ss-therm-res}
\end{align}
where $\hat{k}$ is the estimate of $k$, $L_V \in \mathbb{R}^{N\times 1}$ and $L_T \in \mathbb{R}^{M\times 1}$ are the observer gains to be designed, and ${\hat{\omega}_V}$ and ${\hat{\omega}_T}$ are the estimate of the uncertainties provided by the learning algorithm. Subsequently, subtracting \eqref{ss-echem-obs}-\eqref{ss-echem-obs-2} from \eqref{ss-echem-21}-\eqref{ss-echem-2} and \eqref{ss-therm-21}-\eqref{ss-therm-2} from \eqref{ss-therm-obs}-\eqref{ss-therm-obs-2}, we can write the observers' error dynamics as:
\begin{align}
    & \tilde{z}_{1 t+1} = (A_1-L_VC_1) \tilde{z}_{1t} - L_V (\Delta_{V_{2t}} + {{\epsilon}_V}_t),\label{ss-echem-e}\\
    & r_{Vt} = \tilde{y}_{1t} = C_1 \tilde{z}_{1t} + \Delta_{Vt} + \epsilon_{Vt},\label{ss-echem-e-2}\\
    & \tilde{z}_{2t+1} = (A_2 -L_T C_2) \tilde{z}_{2t} + \tilde{f}_2 + \Delta_{Tt} - L_T \epsilon_{Tt},\label{ss-therm-e}\\
    & r_{Tt} = \tilde{y}_{2t} = C_2 \tilde{z}_{2t} + \epsilon_{Tt},\label{ss-therm-e-2}
\end{align}
where $\tilde{k}=k-\hat{k}$ is the estimation error, $\tilde{f}_2={f}_2({z}_1,y_1,u_1)-{f}_2(\hat{z}_1,y_1,u_1)$, ${{\epsilon}_V}$ and ${{\epsilon}_T}$ represent the error in learned uncertainties.

Next, we present the following proposition that illustrates the convergence properties of the error dynamics and residual signals as well as design conditions for the observer gains $L_V$ and $L_T$.

\begin{prop}
Considering the estimation error dynamics \eqref{ss-echem-e}-\eqref{ss-therm-e-2}, the following are true:
\begin{enumerate}
    \item {in the presence of no fault and no learning error, i.e. ${\Delta_{V_1}}=0,{\Delta_{V_2}}=0,{\Delta_{T}}=0,{\epsilon}_V=0,{\epsilon}_T=0$, the estimation errors ${\tilde{z}_1}$ and ${\tilde{z}_2}$, and the residual signals $r_V$ and $r_T$ will asymptotically converge to zero starting from any non-zero initial condition};
    \item in the presence of fault and learning error, i.e. ${\Delta_{V_1}} \neq 0,{\Delta_{V_2}} \neq 0, {\Delta_{T}} \neq 0, {\epsilon}_V \neq 0, {\epsilon}_T \neq 0$, the estimation errors ${\tilde{z}_1}$ and ${\tilde{z}_2}$, and the residual signals $r_V$ and $r_T$ will remain uniformly bounded;
\end{enumerate}
if there exist symmetric positive definite matrices $P_1$ and $P_2$ such that the following conditions are satisfied:
\begin{align}
& \underline{\lambda}_Q + \Gamma x_1 < 0, \ \ \underline{\lambda}_Z + \vartheta x_2 < 0. \label{cond-1}
\end{align}
where $\Gamma$ and $\vartheta$ are arbitrary positive numbers, $x_1 = \left\|(A_1-L_VC_1)^TP_1\right\|$, $x_2 = \left\|(A_2-L_TC_2)^TP_2\right\|$, $\underline{\lambda}_Q$ and $\underline{\lambda}_Z$ are the minimum eigen values of $[(A_1-L_VC_1)^TP_1(A_1-L_VC_1)-P_1]$ and $[(A_2-L_TC_2)^TP_2(A_2-L_TC_2)-P_2]$, respectively.
\end{prop}

\begin{proof}
See Appendix.
\end{proof}

\subsection{Design of Decision Maker}

The decision maker decides whether a fault has occurred based on the following detection logic:
\begin{align}
    r_i \leq \delta_i \implies \text{no fault}, r_i > \delta_i \implies \text{fault occurrence},
\end{align}
where $i \in \{V,T\}$ and $\delta_i$ are the thresholds. Note that even with learning algorithm there is always a possibility that the residual signals will be non-zero even under no fault conditions. The thresholds $\delta_i$ deal with this issue and provide robustness to such errors. {Next, the following steps are performed for threshold selection.}

{\textit{Step 1:} We run the learning-based detection observers under different no-fault operating conditions. Typically, such operating scenarios are generated using Monte-Carlo type of simulation studies or experimental studies. As we focus on fast charging scenarios in this work, we run the observers under experimental fast charging cycles. Such runs produce residual signal data $r_i$ under no-fault but uncertain conditions. Referring to \eqref{ss-echem-21}-\eqref{ss-therm-2} and \eqref{ss-echem-obs}-\eqref{ss-therm-e-2}, these conditions mean $\omega_i \neq 0$ and $\Delta_i = 0$. Furthermore, considering \eqref{ss-echem-e-2} and \eqref{ss-therm-e-2}, the amplitudes of $r_i$ depend on the amplitudes of the learning error $\epsilon_i$, (i.e. $r_i(\epsilon_i)$).}

{\textit{Step 2:} Next, we find the maximum absolute amplitude of the residual data $r_i$ collected in \textit{Step 1} and set that value as the threshold $\delta_i$, that is $\delta_i = \max \left|r_i\right|$. Effectively, such threshold is related to the maximum learning error in the learning-based detection observers. In other words, the threshold approximately equals to the upper bound of the residuals under uncertain but no-fault conditions, assuming the data $r_i$ comprehensively capture no-fault operating scenarios.}

\section{Results and Discussion}
In the following subsections, we discuss the results of the proposed fault detection scheme. All experiments are conducted on Arbin battery testing system (LBT21084). First, we identify the nominal electrochemical-thermal model for a commercial 18650 Lithium-ion battery cell based on experimental current, terminal voltage, and surface temperature data. The cell has the following characteristics: Graphite anode and NMC cathode, 3 $Ah$ nominal capacity, terminal voltage range 4.2-2.5 V, and maximum fast charge current 4 $A$. Essentially, we have solved the following optimization problem to identify the cell parameters: $\underset{\theta}{\min} \ rms\{ X_e - X_m(\theta)\}$ with subject to \eqref{termV}, \eqref{ss-echem}, \eqref{heat}, and \eqref{ss-therm} 
where $X_e$ and $X_m$ denote experimental and model voltage and temperature data, respectively, and $\theta = \{ D, A_a, A_c, R_b, m_{Li}, \epsilon_a, \epsilon_c, h, C_p, k\}$ is the parameter vector identified as 	$D$=$1.022\times 10^{-14}$ $m^2/s$, $A_a$=$0.09$ $m^2$, $A_c$=$0.048$ $m^2$,$\epsilon_a$=$0.8$, $\epsilon_c$=$0.6516$, $R_b$=$0.006$ $\Omega$, $m_{Li}$=$0.1796$ $moles$, $h$=$16.78$ $W/(m^2-K)$, $C_p$=$907$ $J/(kg-K)$, $k$=$1.79$ $W/(m-K)$. 
{We have used this identified model for the subsequent case studies. The learning algorithm is also applied to learn the uncertainty models. In this particular case, we have used battery data from cycle \# 1 through cycle \# 40 under constant current constant voltage (CCCV) charging to train the uncertainty models. A comparison of the model and experimental data under CCCV charging is shown in Fig. \ref{fig2}.}


\begin{figure}
    \centering
     \includegraphics[trim = 0mm 5mm 0mm 0mm, clip, scale=0.8, width=0.95\linewidth]{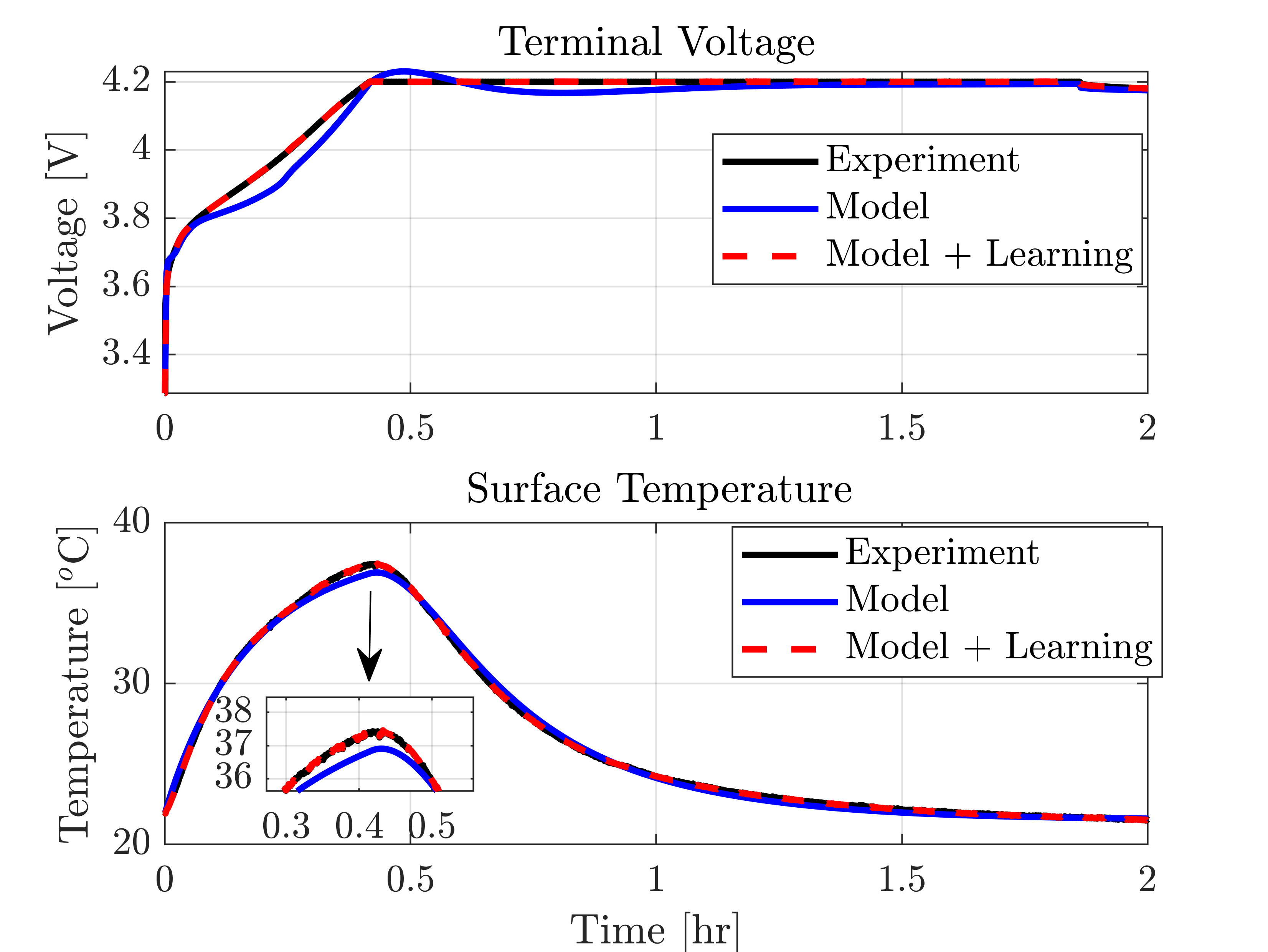}
    \caption{Comparison of model output and experimental data under constant current constant voltage (CCCV) charging scenario. For model only case, the RMS errors are $27.4$ $mV$ and $0.3$ $^oC$ whereas model along with learning have RMS errors of $0.56$ $mV$ and $0.012$ $^oC$.}
    \label{fig2}
\end{figure}

\subsection{Voltage Fault Detection}
In this section, we test the performance of the proposed fault diagnosis scheme under voltage fault. As mentioned before, a voltage fault can be caused by different internal anomalies. In this study, we mainly focus on abrupt type voltage fault caused by Lithium plating which is a highly probably fault mode under fast charging \cite{tomaszewska2019lithium,ahmed2017enabling}. However, the proposed approach is applicable to other type of voltage faults as well. To emulate the fault, we have modified the nominal terminal voltage by adding a faulty voltage component mimicking the voltage data under Lithium plating presented in \cite{yang2018look}. As shown in \cite{yang2018look}, the voltage plateau in lower State-of-Charge region exhibits a \textit{overshoot} type behavior under Lithium plating during high current charging, which is not present in the absence of plating. Experimental voltage data during high current charging and corresponding prediction of Lithium plating is shown in Fig. 1 of \cite{yang2018look}. We incorporate similar \textit{overshoot} type behavior in our voltage fault injection. The faulty voltage component is modelled as: $a_1e^{-0.5(x_1-\mu)}+a_2\sin(x_2)$ where $x_1\in[-\pi,\pi], x_2\in[\frac{\pi}{3},\pi]$ where $a_1$ and $a_2$ represent the fault magnitude. Subsequently, such faulty voltage data was fed to the electrochemical detection observer to test its performance. {The threshold has been designed to be $\delta_V=0.01$ V following the process mentioned in Section III.C.} We have tested four fault cases that capture various magnitudes: Fault case 1: $a_1 =0.003 $ and $a_2 = 0.0075$, Fault case 2: $a_1 = 0.0048$ and $a_2 = 0.0120$, Fault case 3: $a_1 = 0.009$ and $a_2 = 0.0225$. The responses of the terminal voltage ($V_t$) and voltage residual ($r_V$) under these faults are shown in Fig. \ref{fig4B}. {It is observed that all the fault cases are detected as the residuals crossed the threshold in 89, 70 and 49 seconds, respectively. However, Fault case 1 found to be the minimum fault size that can be detected. This illustrates the effectiveness of the proposed scheme in detecting voltage faults.}

\begin{figure}
    \centering
    \includegraphics[trim = 0mm 0mm 0mm 0mm, clip, scale=0.8, width=0.95\linewidth]{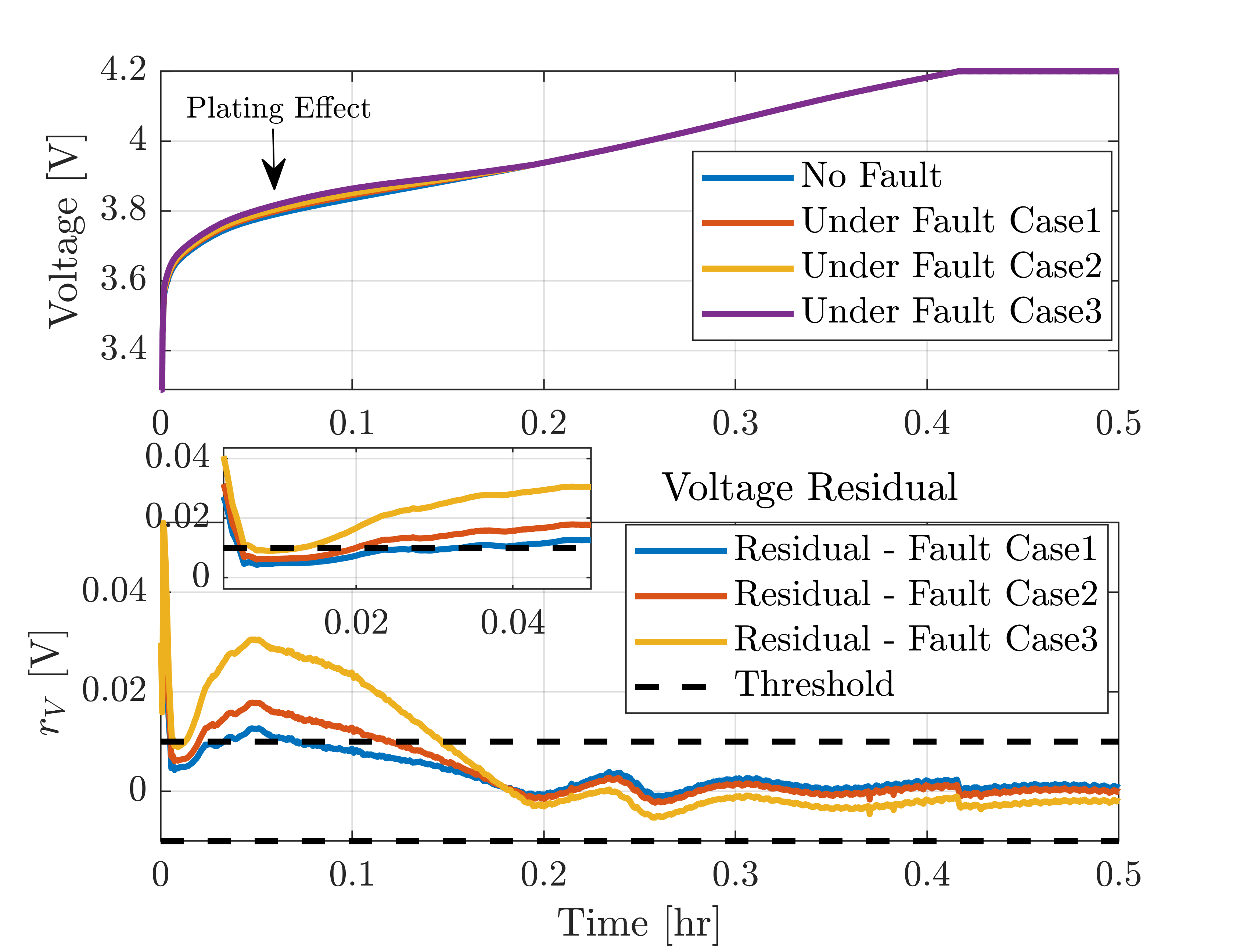}
    \caption{Residual responses under voltage faults.}
    \label{fig4B}
\end{figure}

\subsection{Thermal Fault Detection}
In this section, we discuss an experimental case study under thermal fault. {To emulate the thermal fault scenario in experimental setting, we have used an external heater to inject heat on the surface of the battery cell. In response to such injected heat, battery temperature increased until the external heat has been turned off.} The surface temperature response under such external heat is shown in the top plot of Fig. \ref{fig6}. This temperature data is fed to the proposed fault detection algorithm in order to verify its performance. {The threshold has been designed to be $\delta_T=0.5^o$C following the process mentioned in Section III.C.} As shown in the top plot of Fig. \ref{fig6}, the thermal fault was injected at $311$ s. In response to the injected thermal fault, the residual signal $r_T$ crossed the threshold at $315$ s as shown in the bottom plot of Fig. \ref{fig6}. Hence, the thermal fault was detected within $4$ seconds of its occurrence.

\begin{figure}
    \centering
    \includegraphics[trim = 0mm 0mm 0mm 0mm, clip, scale=0.8, width=0.95\linewidth]{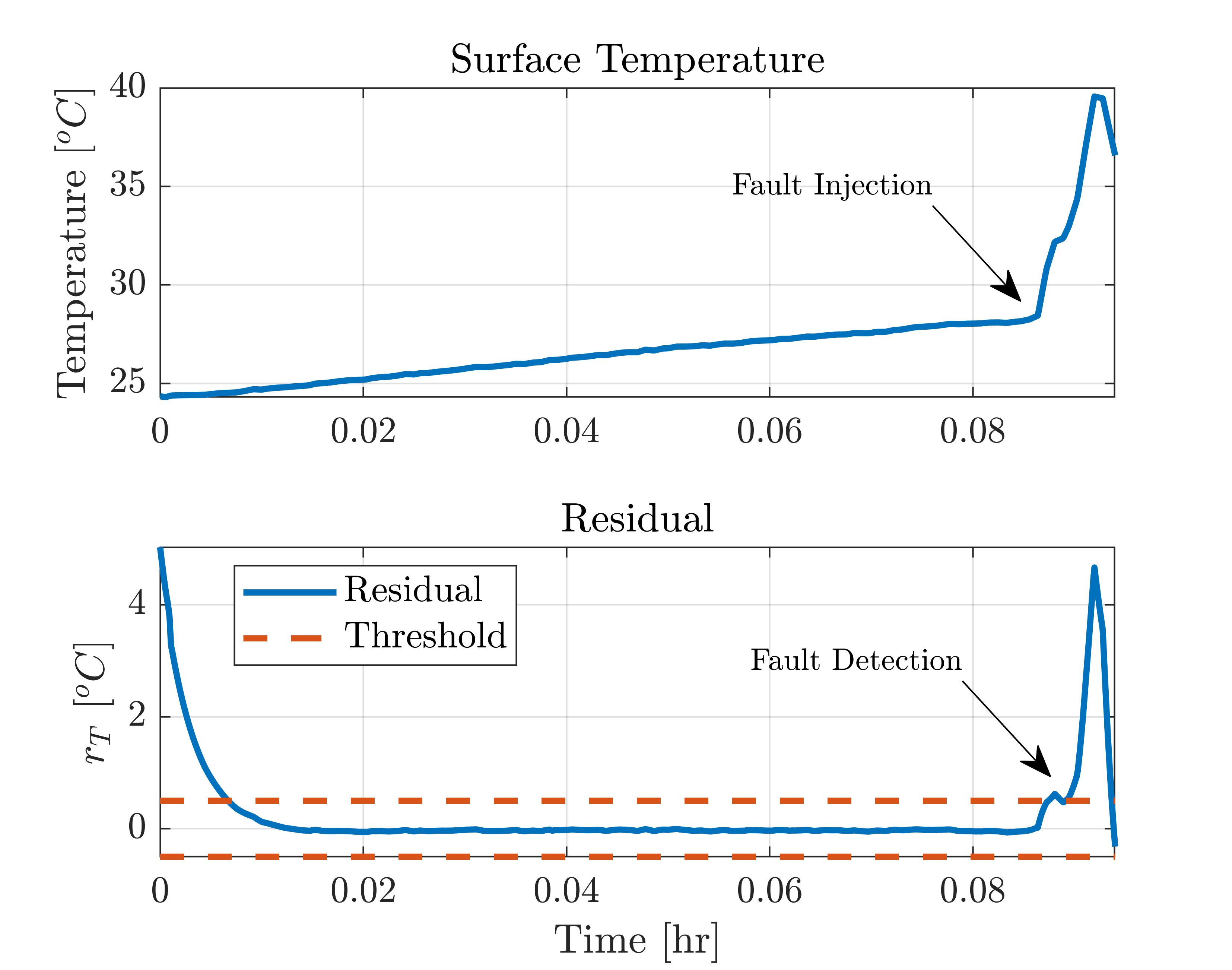}
    \caption{Residual response under thermal fault.}
    \label{fig6}
\end{figure}

\subsection{{Minimum Detectable Thermal Fault}}
{In this study, we explore the minimum size of the thermal fault detectable by the proposed scheme. We have simulated thermal faults of various magnitudes and fed this data to the proposed scheme. We have considered three cases: Case 1: amplitude 310 W, Case 2: amplitude 220 W, Case 3: amplitude 100 W. The surface temperature response and residual responses under these faults are shown in Fig. \ref{fig7}. We observed that: (i) the Case 1 and Case 2 faults are detected within 4 seconds, and (ii) the minimum detectable fault size is around 200 W which is slightly less than the fault in Case 2.}


\begin{figure}
    \centering
    \includegraphics[trim = 0mm 0mm 0mm 0mm, clip, scale=0.8, width=0.95\linewidth]{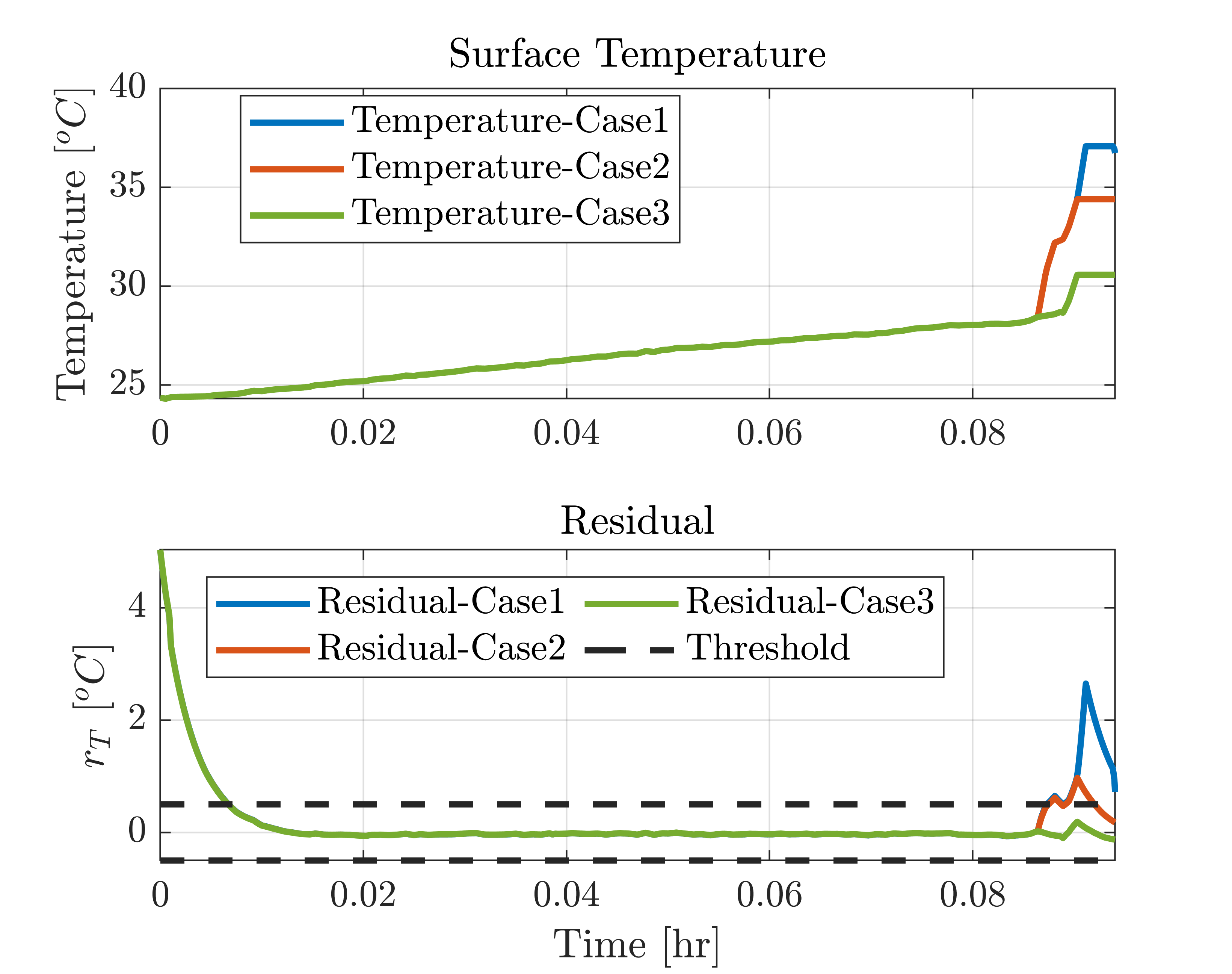}
    \caption{Thermal fault detection performance.}
    \label{fig7}
\end{figure}




\subsection{{Performance Comparison of Model-only and Model-and-Learning Approaches}}
{In this section, we compare the performance of the scheme with and without learning. In the first comparative study, we consider two cases under voltage fault (Fault case 3 in Section III.A): (i) \textit{Case 1}, where the learning algorithm is used in conjunction with the detection observer, and (ii) \textit{Case 2}, where only the detection observer has been used without any learning algorithm. The detection observer in \textit{Case 2} has been designed following standard eigen-value placement based approach. {For both of these cases, there is a \textit{decision maker} which uses the same threshold based approach mentioned in Section III.C. It should be noted that the threshold for \textit{Case 2} (0.075 V) is higher than that of \textit{Case 1} (0.01 V).} This is expected since in \textit{Case 2} the residual signal is corrupted with significantly more amount of uncertainties as there is no learning algorithm. On the other hand, the residual signal in \textit{Case 1} is much less affected by uncertainties due to learning. The residual responses under no fault condition and faulty condition are shown in Fig. \ref{fig5}. We can see that the residual crosses the threshold after the terminal voltage started showing faulty behavior for \textit{Case 1}. However, the residual signal does not cross the threshold for \textit{Case 2}. This shows that the use of learning algorithm can potentially detect smaller faults which would be undetected by without learning based approaches.} {Note that the gain in performance comes at the cost of training and implementation requirement of the learning algorithm. Such implementation would require additional tuning of the hyper-parameters and computational power to process the learning.}

\begin{figure}
    \centering
    \includegraphics[trim = 0mm 0mm 0mm 0mm, clip, scale=0.8, width=0.95\linewidth]{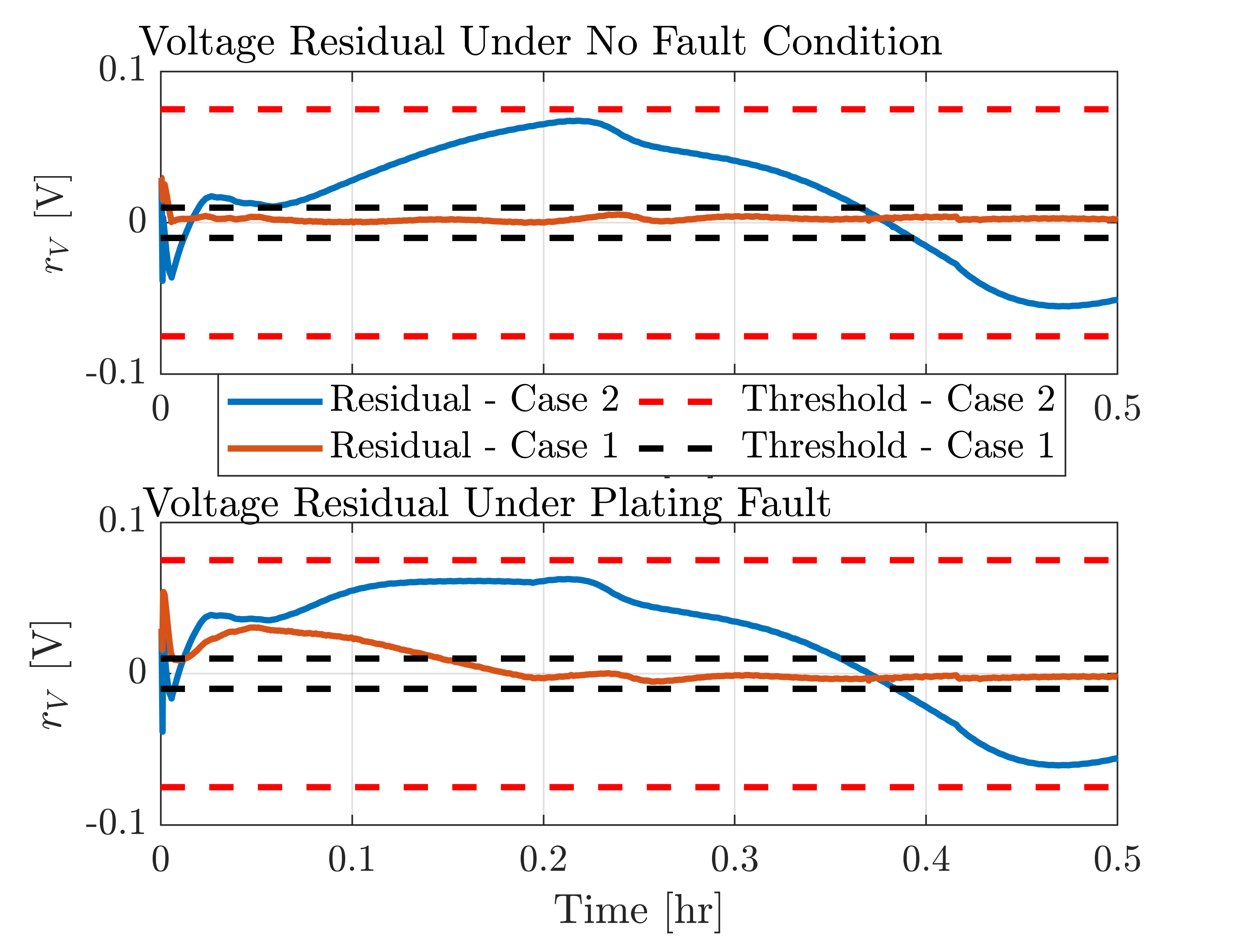}
    \caption{Comparison of with and without learning based detection approaches under Lithium plating fault.}
    \label{fig5}
\end{figure}

\subsection{{Performance under Uncertainties}}
{In this section, the robustness of proposed scheme is studied under parametric uncertainties. Specifically, we illustrate how parametric uncertainties can lead to false alarms. We study two cases: (i) performance of the electrochemical observer with uncertainty in the parameter $a_a$; (ii) performance of the temperature observer with uncertainty in the parameter $R_b$. The results are shown in Fig. \ref{fig4D}. We observed that the performance degrades in terms of false alarm beyond $10\%$ uncertainty in $a_a$ and $20\%$ uncertainty in $R_b$.}

\begin{figure}
    \centering
    \includegraphics[trim = 0mm 0mm 0mm 0mm, clip, scale=0.8, width=0.95\linewidth]{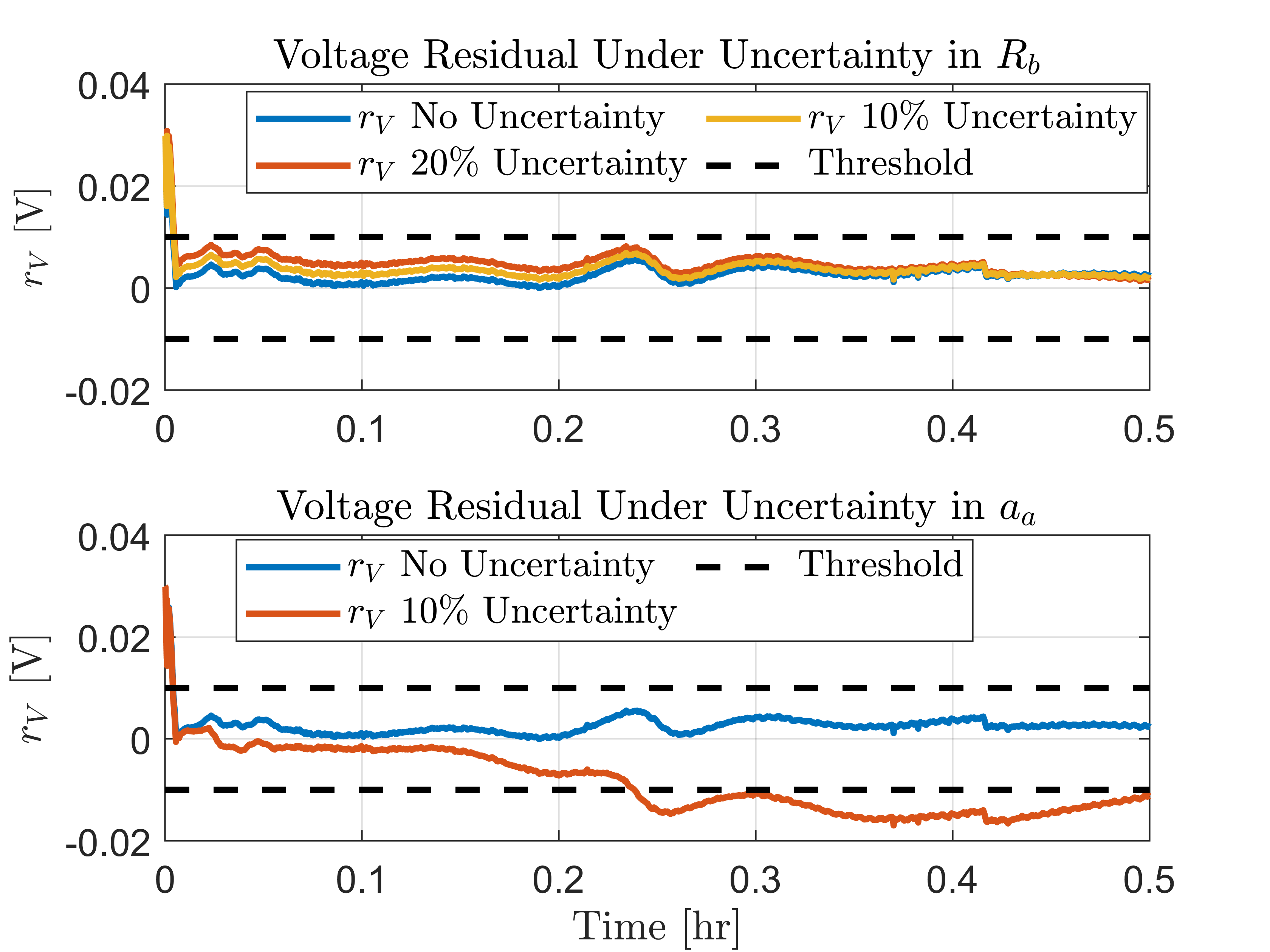}
    \caption{Voltage residuals under uncertainties in $a_a$ and $R_b$.}
    \label{fig4D}
\end{figure}

\section{Conclusion}
In this paper, we have studied a battery fault detection paradigm combining physics-based models with detection observer technique and Gaussian Process regression to detect voltage and thermal faults. Based on an experimentally identified electrochemical-thermal model, we have performed the following case studies: simulation study on voltage fault detection and experimental study on thermal fault detection. These case studies have shown the effectiveness of the proposed detection approach for battery fault detection. We plan to study future extensions including fault estimation, more comprehensive physics-based models, and battery packs. Furthermore, we also plan to improve threshold design approach utilizing comprehensive battery data under faulty and non-faulty conditions.

\appendix[Proof of Proposition 1]
First, we consider the following Lyapunov function candidate to analyze electrochemical detection observer error dynamics $W_{1t} = \tilde{z}_{1t}^T P_1 \tilde{z}_{1t}$ where $P_1$ is a positive definite symmetric matrix. The first order difference in the Lyapunov function is given by:
\begin{align}
     \Delta{W}_1 = & \tilde{z}_{1t}^T [(A_1-L_{VC1})^{T}P_1(A_1-L_{VC1})-P_1] \tilde{z}_{1t} \nonumber\\
    & + 2\tilde{z}_{1t}^T(A_1-L_{VC1})^{T} P_1 \eta_{1t} + \eta_{1t}^{T} P_1 \eta_{1t}. \label{lyap3}
\end{align}
where ${\eta_1} = -L_V ({\Delta_V} + \epsilon_V)$. Considering Holder's inequality and the inequality $2ab \leq \Gamma a^2 + \frac{1}{\Gamma}b^2$ with $a,b,\Gamma >0$, we can further re-write \eqref{lyap3} as
\begin{align}
     \Delta{W}_1 \leq & (\underline{\lambda}_Q + \Gamma x_1) \left\|\tilde{z}_{1t}\right\|^2 + (\overline{\lambda}_P + \frac{x_1}{\Gamma}) \left\|\eta_{1t}\right\|^2, \label{lyap4}
\end{align}
where $\Gamma$ is an arbitrary positive number, $$x_1 = \left\|(A_1-L_{VC1})^TP_1\right\|$$, $\underline{\lambda}_Q$ is the minimum eigen value of $[(A_1-L_{VC1})^TP_1(A_1-L_{VC1})-P_1]$, $\overline{\lambda}_P$ is the maximum eigen value of $P_1$. If first equation of \eqref{cond-1} is satisfied, then we have $\underline{\lambda}_Q + \Gamma x_1 < 0$. Consequently, we have $$\Delta{W}_1 <0$$ when $${\Delta_{V}}=0,\Delta_{T}=0,{\epsilon}_V=0,{\epsilon}_T=0$$. Hence, the estimation error $\tilde{z}_1$ and the residual signal $r_V$ will asymptotically converge to zero starting from a non-zero initial condition.

When $\Delta_{V} \neq 0, \Delta_{T} \neq 0, \epsilon_V \neq 0, \epsilon_T \neq 0$, then we can only guarantee $\Delta{W}_1 <0$ under the condition $\left\|\tilde{z}_{1t}\right\|^2 > -\frac{(\overline{\lambda}_P + \frac{x_1}{\Gamma})}{(\underline{\lambda}_Q + \Gamma x_1)}\left\| \eta_{1t}\right\|^2$. Hence, we can conclude that the estimation error $\tilde{z}_1$ and the residual signal $r_V$ will be uniformly bounded and converge within a ball of radius. Following similar steps, we can prove the convergence for the thermal detection observer.


%

%
%
%


\bibliographystyle{IEEEtran}
\bibliography{ref1}\ 

\begin{thebibliography}{10}
\providecommand{\url}[1]{#1}
\csname url@samestyle\endcsname
\providecommand{\newblock}{\relax}
\providecommand{\bibinfo}[2]{#2}
\providecommand{\BIBentrySTDinterwordspacing}{\spaceskip=0pt\relax}
\providecommand{\BIBentryALTinterwordstretchfactor}{4}
\providecommand{\BIBentryALTinterwordspacing}{\spaceskip=\fontdimen2\font plus
\BIBentryALTinterwordstretchfactor\fontdimen3\font minus
  \fontdimen4\font\relax}
\providecommand{\BIBforeignlanguage}[2]{{%
\expandafter\ifx\csname l@#1\endcsname\relax
\typeout{** WARNING: IEEEtran.bst: No hyphenation pattern has been}%
\typeout{** loaded for the language `#1'. Using the pattern for}%
\typeout{** the default language instead.}%
\else
\language=\csname l@#1\endcsname
\fi
#2}}
\providecommand{\BIBdecl}{\relax}
\BIBdecl

\bibitem{ahmed2017enabling}
S.~Ahmed, I.~Bloom, A.~N. Jansen, T.~Tanim, E.~J. Dufek, A.~Pesaran,
  A.~Burnham, R.~B. Carlson, F.~Dias, K.~Hardy \emph{et~al.}, ``{Enabling fast
  charging--A battery technology gap assessment},'' \emph{Journal of Power
  Sources}, vol. 367, pp. 250--262, 2017.

\bibitem{dey2017model}
S.~Dey, H.~E. Perez, and S.~J. Moura, ``Model-based battery thermal fault
  diagnostics: Algorithms, analysis, and experiments,'' \emph{IEEE Transactions
  on Control Systems Technology}, vol.~27, no.~2, pp. 576--587, 2017.

\bibitem{dey2015diagnostic}
S.~Dey and B.~Ayalew, ``A diagnostic scheme for detection, isolation and
  estimation of electrochemical faults in lithium-ion cells,'' in \emph{ASME
  2015 Dynamic Systems and Control Conference}.\hskip 1em plus 0.5em minus
  0.4em\relax American Society of Mechanical Engineers Digital Collection,
  2015.

\bibitem{marcicki2010nonlinear}
J.~Marcicki, S.~Onori, and G.~Rizzoni, ``Nonlinear fault detection and
  isolation for a lithium-ion battery management system,'' in \emph{ASME 2010
  Dynamic Systems and Control Conference}.\hskip 1em plus 0.5em minus
  0.4em\relax American Society of Mechanical Engineers Digital Collection,
  2010, pp. 607--614.

\bibitem{wei2019lyapunov}
J.~Wei, G.~Dong, and Z.~Chen, ``Lyapunov-based thermal fault diagnosis of
  cylindrical lithium-ion batteries,'' \emph{IEEE Transactions on Industrial
  Electronics}, 2019.

\bibitem{xiong2019online}
R.~Xiong, R.~Yang, Z.~Chen, W.~Shen, and F.~Sun, ``Online fault diagnosis of
  external short circuit for lithium-ion battery pack,'' \emph{IEEE
  Transactions on Industrial Electronics}, 2019.

\bibitem{saha2008uncertainty}
B.~Saha and K.~Goebel, ``Uncertainty management for diagnostics and prognostics
  of batteries using bayesian techniques,'' in \emph{2008 IEEE Aerospace
  Conference}.\hskip 1em plus 0.5em minus 0.4em\relax IEEE, 2008, pp. 1--8.

\bibitem{zhao2017fault}
Y.~Zhao, P.~Liu, Z.~Wang, L.~Zhang, and J.~Hong, ``Fault and defect diagnosis
  of battery for electric vehicles based on big data analysis methods,''
  \emph{Applied Energy}, vol. 207, pp. 354--362, 2017.

\bibitem{zhang2018long}
Y.~Zhang, R.~Xiong, H.~He, and M.~G. Pecht, ``Long short-term memory recurrent
  neural network for remaining useful life prediction of lithium-ion
  batteries,'' \emph{IEEE Transactions on Vehicular Technology}, vol.~67,
  no.~7, pp. 5695--5705, 2018.

\bibitem{li2020battery}
D.~Li, Z.~Zhang, P.~Liu, Z.~Wang, and L.~Zhang, ``{Battery Fault Diagnosis for
  Electric Vehicles Based on Voltage Abnormality by Combining the Long
  Short-Term Memory Neural Network and the Equivalent Circuit Model},''
  \emph{IEEE Transactions on Power Electronics}, vol.~36, no.~2, pp.
  1303--1315, 2020.

\bibitem{lee2021convolutional}
H.~Lee, K.~Kim, J.-H. Park, G.~Bere, J.~J. Ochoa, and T.~Kim, ``Convolutional
  neural network-based false battery data detection and classification for
  battery energy storage systems,'' \emph{IEEE Transactions on Energy
  Conversion}, 2021.

\bibitem{berkenkamp2015safe}
F.~Berkenkamp and A.~P. Schoellig, ``Safe and robust learning control with
  gaussian processes,'' in \emph{2015 European Control Conference (ECC)}.\hskip
  1em plus 0.5em minus 0.4em\relax IEEE, 2015, pp. 2496--2501.

\bibitem{recht2019tour}
B.~Recht, ``A tour of reinforcement learning: The view from continuous
  control,'' \emph{Annual Review of Control, Robotics, and Autonomous Systems},
  vol.~2, pp. 253--279, 2019.

\bibitem{rasmussen2003gaussian}
C.~E. Rasmussen, ``Gaussian processes in machine learning,'' in \emph{Summer
  School on Machine Learning}.\hskip 1em plus 0.5em minus 0.4em\relax Springer,
  2003, pp. 63--71.

\bibitem{santhanagopalan2006online}
S.~Santhanagopalan and R.~E. White, ``Online estimation of the state of charge
  of a lithium ion cell,'' \emph{Journal of power sources}, vol. 161, no.~2,
  pp. 1346--1355, 2006.

\bibitem{7004795}
S.~{Dey}, B.~{Ayalew}, and P.~{Pisu}, ``Nonlinear robust observers for
  state-of-charge estimation of lithium-ion cells based on a reduced
  electrochemical model,'' \emph{IEEE Transactions on Control Systems
  Technology}, vol.~23, no.~5, pp. 1935--1942, 2015.

\bibitem{ALHALLAJ19991}
S.~A. Hallaj, H.~Maleki, J.~Hong, and J.~Selman, ``Thermal modeling and design
  considerations of lithium-ion batteries,'' \emph{Journal of Power Sources},
  vol.~83, no.~1, pp. 1 -- 8, 1999.

\bibitem{guo2011thermal}
M.~Guo and R.~E. White, ``Thermal model for lithium ion battery pack with mixed
  parallel and series configuration,'' \emph{Journal of the Electrochemical
  Society}, vol. 158, no.~10, p. A1166, 2011.

\bibitem{hu2020advanced}
X.~Hu, K.~Zhang, K.~Liu, X.~Lin, S.~Dey, and S.~Onori, ``{Advanced Fault
  Diagnosis for Lithium-Ion Battery Systems: A Review of Fault Mechanisms,
  Fault Features, and Diagnosis Procedures},'' \emph{IEEE Industrial
  Electronics Magazine}, vol.~14, no.~3, pp. 65--91, 2020.

\bibitem{tomaszewska2019lithium}
A.~Tomaszewska, Z.~Chu, X.~Feng, S.~O'Kane, X.~Liu, J.~Chen, C.~Ji, E.~Endler,
  R.~Li, L.~Liu \emph{et~al.}, ``{Lithium-ion battery fast charging: a
  review},'' \emph{ETransportation}, vol.~1, p. 100011, 2019.

\bibitem{doughty2012general}
D.~H. Doughty and E.~P. Roth, ``A general discussion of li ion battery
  safety,'' \emph{The Electrochemical Society Interface}, vol.~21, no.~2, pp.
  37--44, 2012.

\bibitem{bandhauer2011critical}
T.~M. Bandhauer, S.~Garimella, and T.~F. Fuller, ``A critical review of thermal
  issues in lithium-ion batteries,'' \emph{Journal of the Electrochemical
  Society}, vol. 158, no.~3, pp. R1--R25, 2011.

\bibitem{3569}
C.~Rasmussen and C.~Williams, \emph{Gaussian Processes for Machine Learning},
  ser. Adaptive Computation and Machine Learning.\hskip 1em plus 0.5em minus
  0.4em\relax Cambridge, MA, USA: MIT Press, Jan. 2006.

\bibitem{yang2018look}
X.-G. Yang, S.~Ge, T.~Liu, Y.~Leng, and C.-Y. Wang, ``A look into the voltage
  plateau signal for detection and quantification of lithium plating in
  lithium-ion cells,'' \emph{Journal of Power Sources}, vol. 395, pp. 251--261,
  2018.

\end{thebibliography}

\end{document}